\documentclass[fleqn,10pt]{wlscirep}

\usepackage{graphicx}
\usepackage{dcolumn}
\usepackage{bm}
\usepackage{multirow}

\usepackage{listings}

\usepackage{amsmath,amssymb,amsfonts}
\usepackage{graphicx}
\usepackage{color}
\usepackage{appendix}
\usepackage{amsthm}
\usepackage{tikz}
\usepackage{subfig}

\newtheorem{theorem}{Theorem}

\newtheorem{lemma}{Lemma}
\newtheorem{definition}{Definition}

\newcommand{\be}{\begin{eqnarray}}
\newcommand{\ee}{\end{eqnarray}}
\newcommand{\nn}{\nonumber}
\newcommand{\bn}{\begin{enumerate}}
	\newcommand{\en}{\end{enumerate}}
\newcommand{\bl}{\begin{align}}
\newcommand{\el}{\end{align}}

\def\a{\alpha}

\def\d{\delta}

\def\r{\rho}

\def\t{\tau}

\def\p{\psi}

\def\P{\Psi}

\def\O{\Omega}

\def\iff{\Longleftrightarrow}

\def\<{\langle}
\def\>{\rangle}

\def\cT{{\mathcal{T}}}
\def\cO{{\mathcal{O}}} 
\def\cE{{\mathcal{E}}} 
\def\vn{{\vec{n}}}
\def\vm{{\vec{m}}}
\def\hU{{\hat{U}}}

\def\ha{{\hat{a}}}

\def\bra#1{\mathinner{\langle{#1}|}}
\def\ket#1{\mathinner{|{#1}\rangle}}

\title{Generalized concurrence in boson sampling}

\author[1]{Seungbeom Chin}
\author[1,*]{Joonsuk Huh}
\affil[1]{Department of Chemistry, Sungkyunkwan University, Suwon 16419, Korea}

\affil[*]{joonsukhuh@skku.edu}

\begin{abstract}

A fundamental question in linear optical quantum computing is to understand the origin of the quantum supremacy in the physical system. It is found that the multimode linear optical transition amplitudes are calculated through the permanents of transition operator matrices, which is a hard problem for classical simulations (boson sampling problem).
We can understand this problem by considering a quantum measure that directly determines the runtime for computing the transition amplitudes.
In this paper, we suggest a quantum measure named ``Fock state concurrence sum'' $C_S$, which is the summation over all the members of ``the generalized Fock state concurrence'' (a measure analogous to the generalized concurrences of entanglement and coherence).
By introducing generalized algorithms for computing the transition amplitudes of the Fock state boson sampling with an arbitrary number of photons per mode, we show that the minimal classical runtime for all the  known algorithms directly depends on $C_S$. Therefore, we can state that \emph{the Fock state concurrence sum $C_S$ behaves as a collective measure that controls the computational complexity of Fock state BS}.
We expect that our observation on the role of the Fock state concurrence in the generalized algorithm for permanents would provide a unified viewpoint to interpret the quantum computing power of linear optics.    	
\end{abstract}

\begin{document}

\flushbottom
\maketitle
%
%
\thispagestyle{empty}

\section*{Introduction}
The extended Church-Turing thesis (ECT) states that every problem that can be efficiently computable with real physical devices are efficiently simulated with a Turing machine. It is expected that quantum computers would refute ECT by exploiting its inherent quantum supremacy. However, since scalable universal quantum computers that can perform actual quantum algorithms are not likely to be built in the foreseeable future, they are not ``real physical devices'' yet.  
 
Boson sampling (BS) \cite{Aaronson2011} was introduced to defeat ECT with more feasible quantum devices, i.e., the linear optical network (LON) implementation. BS is considered a non-universal quantum computer with multi-photons in the multimode optical network.
Aaronson and Arkhipov~\cite{Aaronson2011} claimed that the transition amplitudes with no more than one photon per mode becomes hard to simulate with classical computers as the system scale increases.

The computational hardness of BS is from the hardness of matrix permanents.
The transition amplitude from a pre-selected input state to a post-selected output state is determined by the permanent of a submatrix of a unitary matrix $U$ in the LON. When no more than one photon is in both all input and output modes of the system,  the amplitude can be classically simulated with Ryser's formula \cite{Ryser1963} (the best known algorithm for computing permanents).
We recall the definition of the permanent of an $M$-dimensional square matrix $A$ (Per[$A$]):
\begin{align}
\mathrm{Per}[A]=\sum_{\vec{{\sigma} }\in S}\prod_{i=1}^{M}A_{i,\sigma_{i}} ,
\label{eq:PERA}
\end{align}
where $A_{ij}$ are the entries of $A$ and the set $S$ includes all the  permutations of $(1,2,\ldots,M)$, $\{\vec{\sigma} \}$ ($\vec{\sigma} =(\sigma_1,\ldots \sigma_M)$ is an $N$-dimensional vector).  
The brute force computation of a matrix permanent in Eq. \eqref{eq:PERA} requires $N!$ terms in summation and each term is composed of the products of  $N$ elements of the matrix. Even though  Ryser's algorithm~\cite{Ryser1963} can perform the calculation in $O(2^{N-1}N^{2})$ arithmetic operations (it can be optimized further by Gray code as $O(2^{N-1}N)$ operations), the number of operations still increases exponentially with $N$ (it was shown in Valiant (1979) \cite{valiant} and Aaronson (2011) \cite{aaronson2011b} that the computation of permanent is a $\#$P-hard problem).
Glynn~\cite{Glynn2010,Glynn2013} derived a different algorithm that has the same order of computational cost with that of Ryser's. Even though Jerrum et al. \cite{jerrum2004} suggested a polynomial-time approximation algorithm for the permanents of matrices with non-negative elements, there exists no algorithms for arbitrary matrices that are more efficient than Ryser's and 
Glynn's yet. On the other hand, there have been some efforts in developing randomized algorithms for the permanents. Gurvits used Glynn's formula to design a randomized algorithm~\cite{gurvits2005}, and Aaronson and Hance~\cite{aaronson2012} generalized Gurvits's sampling algorithm for matrices with either of repeated columns or repeated rows.  A more generalized algorithm for matrices with repeated rows and columns, which can estimate the complexity of Fock state BS with multiple photons in both input and output modes, is introduced by Yung $et$ $al.$ \cite{yung2016}
When these algorithms are randomized, they estimate the matrix permanent with additive errors in polynomial runtimes.  In this paper, we propose another generalized algorithm for matrices with repeated rows and columns. It is achieved by exploiting a series expansion of a product of variables regarding the linear combinations of variables \cite{kan:2007}.

The classical minimal runtimes ($\cT_{min}$) of the algorithms mentioned above have interesting mathematical features, which render the algorithm to be related to a more general viewpoint of quantum complexity.  
The first observation is that the algorithm we derived here has the same $\cT_{min}$ as that of the formula in Yung $et$ $al.$ \cite{yung2016}, such as Ryser's and Glynn's have the same $\cT_{min}$. Considering the two algorithms arise from very different mathematical structures, we can regard that the obtained runtime is a rigorous criterion for the computational complexity of Fock state BS.
The second obseration is that the functional form of $\cT_{min}$ contains a summation of \emph{elementary symmetric polynomials}. They have an intimate functional relation with the recently introduced coherence monotones, the coherence rank and generalized coherence concurrence \cite{killoran, Chin2017cn, Chin2017gcn}. This motivates us to define \emph{the generalized Fock state concurrence} for a given state $|\vn\>$, which consists of \emph{the Fock state $k$-concurrence} denoted by $C_{k}(\vn)$ with $0 \le k \le N$, and \emph{the Fock state concurrence sum} (the summation of $C_{k}(\vn)$ from $k=0$ to $k=N$ and denoted by $C_S(\vn)$). The Fock state concurrence sum $C_S(\vn)$ is directly related to the amount of $\cT_{\min}$. We can state that the increase of $C_S(\vn)$ results in larger computational complexity, or $C_S(\vn)$ is a quantum resource for the complexity of the given  system.

The concept of Fock state concurrence can also be compared with the Boltzmann entropy of the elementary quantum complexity $S_B^q$ introduced in Chin $et$ $al.$ (2017)\cite{Chin2017maj}, which naturally emerges from the additive error bound for the approximated permanent estimator. By encompassing entropy and concurrence,
our suggestion in this paper would provide the foundation for the quantum resource theory of linear optical quantum computing. In other words, by understanding the role of these quantum measures, we could find the origin of the quantum supremacy in quantum linear optics.

\section*{Results}

First, the generalized Fock state concurrence and the concurrence sum are defined, and their physical relation with the generalized coherence concurrence \cite{Chin2017gcn} is explained. Then an algorithm for computing the transition amplitudes of Fock state BS with multiple photons in input/output modes is proposed. By analyzing the minimal runtime of three algorithms (including ours) for computing the transition amplitudes of Fock state BS, we show that the Fock state concurrence sum is a quantum resource that determines the complexity of a given Fock state BS system.

\subsection*{The generalized Fock state concurrence family and the concurrence sum}\label{concurrence}

Many theoretical analyses support the belief that quantum computers can perform some tasks faster than classical computers. Accordingly, it has been of particular interest to find the resources required for the quantum speedup.
It is believed that entanglement is a critical resource for universal quantum computers \cite{jozsa2003,vidal2003,van2006};  however, the efficiency does not simply depend on the amount of entanglement \cite{van2013}.
It is also recently shown that the original Grover algorithm monotonically consumes coherence during the searching process \cite{shi2017, Chin2017cn}.
There have been attempts to approach the problem in the reverse  direction as well, i.e., to find conditions for a quantum system not to have any speedup. It is rigorously shown  that nonnegative probability quasi-distributions (PQD) result in no quantum speedup \cite{veitch2012, mari2012, veitch2013}.  

In the case of BS,
the photon indistinguisability is considered the origin of the computational complexity in the Fock state BS \cite{Aaronson2011, tichy2015}, and the degree of complexity is closely related to the majorization of the input-output photon distributions \cite{Chin2017maj}. Whereas
Rahimi-Keshari $et.$ $al$ (2016)\cite{rahimi2016} approached this problem from the perspective of quasi-probability distributions (QPD), showing that the negativity of probability quasi-distribution (PQD) of linear optical networks is the necessary resource for the complexity.  

In this section, we define a quantum measure from the multi-photon distribution patterns in multimode optical systems, \emph{the generalized Fock state concurrence} and \emph{the Fock state concurrence sum}. The generalized Fock state concurrence is a quantity analogous to the generalized entanglement concurrence \cite{gour2005} and generalized coherence concurrence \cite{Chin2017gcn}. It will be shown in the later section that the Fock state concurrence sum becomes a resource that determines the complexity of Fock state BS.

\subsubsection*{Definitions}

In the linear optical network of $M$ optical modes into which $N$ photons are injected, the Fock state vector is written as
\begin{align}
|\vn\>=\ket{n_1, n_2, ..., n_M}, \qquad \sum_i^M n_i =N
\end{align} 
where $n_i$ represents the photon number for the $i$th mode (it is worth emphasizing that $n_i$ can be greater than 1 for our later discussion on the generalized Fock state BS). Then the coherence rank and $k$-concurrence of a Fock state $|\vn\>$ is defined as follows:

\begin{definition}\label{rank} The Fock state coherence rank for a given Fock state $\vn$ is defined as
	the integer $ \a_{\vec{n} }$, the number of nonzero elements for the particle distribution vector $\vn$. 
\end{definition}

\begin{definition}\label{kconcurrence}
	The Fock state $k$-concurrence for a given Fock state $\vn$ is defined with the elementary symmetric polynomial as  
	\begin{align}
	C_k(\vec{n}) \equiv \Bigg( \frac{1}{\binom{N}{k} } X_k(\vec{n}) \Bigg)^{\frac{1}{k}},
	\end{align}
	where $X_k(\vn)$ ($k$th elementary symmetric polynomial) is defined as
	\begin{align}
	X_k(\vn) = \sum_{i_1<i_2<\cdots <i_k =1 }^{\a_{\vec{n}}} n_{i_1}n_{i_2}\cdots n_{i_k}, \quad (0 \le k\le\a_{\vn}\le  M)
	\label{esp}
	\end{align}
	(we define $X_0(\vn)=1$, and $X_1(\vn)=N$ for any $\vn$).
\end{definition}
This is normalized so that $C_k(\vn)$ becomes 1 when $\vn$ is maximally coherent, i.e., $\vn=(\vec{1}_{N}, \vec{0}_{M-N})$. The Fock state $k$-concurrences $C_k(\vn)$ from $0\le k \le N$ constitute the generalized Fock state concurrence family.
\begin{definition}\label{concurrencesum}
	The Fock state concurrence sum for a given Fock state $\vn$ is defined as
	\begin{align}
	C_S(\vec{n}) \equiv \frac{1}{2^N}\sum_{k=0}^{\a_\vn} \binom{N}{k}[C_k(\vec{n})]^k,
	\end{align} 
\end{definition}
The factor $1/2^N$ is multiplied for the normalization, i.e., $C(\vn)=1$ when $\vn$ is maximally coherent.
Since $X_k$ are all Schur concave functions,  which decrease as the Fock state vector $\vn$ is more majorized \cite{olkin2}, the $k$-concurrences are also Schur concave functions. Hence the concurrence sum $C_S(\vn)$ is also a Schur concave function.

To calculate the Fock state $k$-concurrence (and concurrence sum) of those states which are not expressed with a single photon distribution vector, we need to consider more comprehensive definitions than Definition \ref{kconcurrence}. It can be achieved in a similar manner to the generalized concurrences of entanglement and coherence \cite{gour2005, Chin2017gcn}, the situation is slightly different for our case though. The well-known convex-roof extention (see, e.g., Eltschka $et$ $al.$ (2014)\cite{elt}) is not exactly suitable here, for Definition \ref{kconcurrence} does not embrace the pure states that are superpositions of photon distribution vectors, i.e., when $|\p\> =\sum_{\vn}\p_\vn |\vn\>$ ($\sum_{\vn}|\p_\vn|^2=1$). Hence, we need two steps of extension for the generalized Fock state concurrence family:

\begin{definition}
	\label{generalconcurrence}
	The Fock state $k$-concurrence of a pure state 
	$|\p\> =\sum_{\vn}\p_\vn |\vn\>$ is defined as 
	\begin{align}
	C_k(|\p\>) \equiv \sum_{\vn}|\p_\vn|^2 C_k(|\vn\>), 
	\end{align}
	and the Fock state concurrence sum of $|\p\>$ as
	\begin{align}
	C_S(|\p\>) \equiv \sum_{\vn}|\p_\vn|^2 C_S(|\vn\>).
	\end{align}
	The Fock state $k$-concurrence of a mixed state $\r$, which can be pure-state-decomposed as $\r= \sum_{a}\r_{a}|\p_a\>\<\p_a|$, is defined with the convex roof extension as 
	\begin{align}
	C_k(\r) \equiv \min_{\{\r_a,|\p_a \> \}} \r_aC_{k}(|\p_a\>),
	\end{align}
	and the Fock state concurrencce sum of $\r$ as
	\begin{align}
	C_S(\r) \equiv \min_{\{\r_a,|\p_a \> \}} \r_aC_{S}(|\p_a\>).
	\end{align}
\end{definition}

\subsubsection*{Comparison of the Fock state concurrence with the single particle coherence concurrence}

We can explain the intuitive relation between the Fock state coherence and the single photon coherence, which will clarify our concept of the Fock state concurrence.

\begin{figure}[htbp]
	\centering
	\subfloat{{\includegraphics[height=4cm]{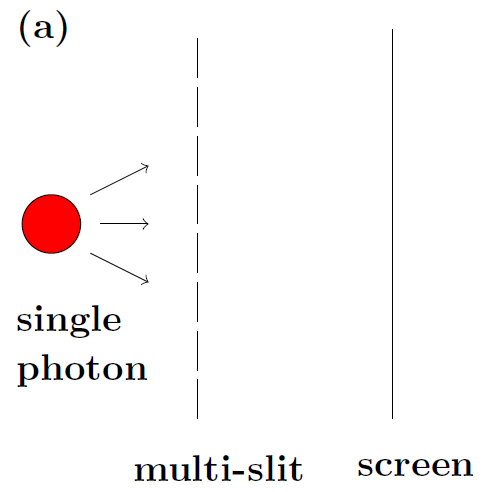} }}%
	\qquad
	\subfloat{{\includegraphics[height=4cm]{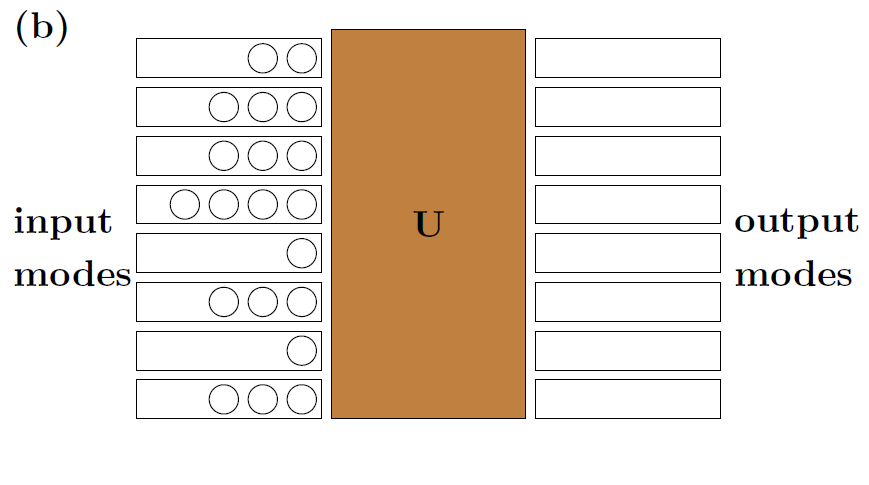} }}%
	\caption{(a) A single photon that passes through a well-separated multi-slit (b) Multi-photons injected to a multi-mode linear optical network system $U$}
\end{figure}

The coherence as one of the fundamental non-classicalities is originated from the framework of superposition \cite{killoran, theurer2017}, i.e., the partition of probability among several states for one quantum system. Since coherence is basis-dependent, we first need to fix a computational basis set. The quantification of coherence is possible under a given normalized basis set $\{|i\> \}_{i=1}^{d}$, and we can state that a pure state is coherent in the basis set if and only if 
\begin{align}
|\p\> = \sum_{i=1}^{k >1}\p_i|i\>.
\end{align}
When $k=1$, $|\p\>$ is incoherent (the mixed state extension of coherence is straightforward. See Baumtratz $et$ $al.$ (2014)\cite{baumgratz2014}). The degree of coherence is determined by the probability amplitude of the state, i.e.,
\begin{align}
P(|\p\>) = (|\p_1|^2, |\p_2|^2, \ldots, |\p_d|^2)\qquad (\sum_{i=1}^d |\p_i|^2=1).
\end{align}
The concept of majorization plays a crucial role here (for two nonincreasing real vectors $\vec{x}$ and $\vec{y}$ of dimension $d$, we state that $\vec{x}$ is majorized by $\vec{y}$ (or $\vec{x} \prec \vec{y}$) if and only if $\sum_{i=1}^k x_i \le \sum_{i=1}^k y_i$ for all $k< d$ and $\sum_{i=1}^d x_i = \sum_{i=1}^d y_i$ \cite{olkin2}). Indeed, for two pure states $|\p\>$ and $|\phi\>$, the following relation holds \cite{du2015}:
\begin{align}
&\textrm{$P(|\p\>)$ is majorized by $P(|\phi\>)$} \nn \\ &\quad \iff\quad  \textrm{$|\p\>$ is more coherent  than $|\phi\>$, i.e., $|\p\>$ can be transformed to $|\phi\>$ with incoherent operations (IO). } 
\label{majcoh}
\end{align} 

Two extremal cases are when $P(|\p\>) = \textrm{perm}[(1,0,\ldots,0)]$ (perm[$\vec{v}$] denotes any permutation vector of $\vec{v}$) and $P(|\p\>) = \frac{1}{\sqrt{d}} (1,1,\ldots, 1)$. The state is incoherent for the former and maximally coherent for the latter. There are several coherence measures that satisfy \eqref{majcoh} (see, e.g., Streltsov $et$ $al.$ (2017)\cite{streltsov2017} for some examples).  

One specific example is the d-slit experiment of a photon (Fig. 1 (a)). When each slit is well-separated from the others, the photon state that passes through the slit is represented in the computational basis set $\{ |\p_1\>, |\p_2\> ,\ldots , |\p_d\> \}$ ($\<\p_i|\p_j\>= \d_{ij}$), where $|\p_i\>$ corresponds to the case when the photon passes the $i$-th slit. Then the photon state is given by
\begin{align}
|\P\> = \sum_{i=1}^d c_i|\p_i\> \qquad \Big( \sum_i|c_i|^2=1 \Big).
\end{align}   
Therefore, we can state that the coherence of $|\p\> $ is determined by 
\begin{align}
P(|\P\>) = (|c_1|^2, |c_2|^2, \ldots, |c_d|^2).
\end{align}
When $|c_i|^2 =1$ for some $i$, the state is incoherent and passes through the $i$-th slit deterministically. This state represents the particle-like property of the photon. On the other hand, when $|c_i|^2=1/\sqrt{d}$ for all $i$, the state is maximally coherent and represents the wave-like property. The analysis can be generalized to the mixed state case by attaching a detector at the slit \cite{bera2015, Chin2017gcn}.
The coherence $k$-concurrence $C_c^{(k)}$ for a pure state  $|\P\>$ is given by \cite{Chin2017gcn}
\begin{align}
\label{slitkconcurrence}
C_c^{(k)}(|\P\>) = d\Bigg( \frac{1}{\binom{d}{k}} \sum_{i_1<i_2<\cdots <i_k =1}^{d} X_k\Big[P(|\P\>)\Big]\Bigg)^\frac{1}{k},
\end{align}   
where $X_k\Big[P(|\P\>)\Big]$ is the $k$-th elementary symmetric polynomial of $P(|\P\>)$.

In the multi-photon case as to the Fock state BS, the probablity for each photon to be in a specific mode is 1. On the other hand, the probability distribution of the initial $N$-photon in $M$-modes are given by 
\begin{align}
P(|\vn\>) = \Big(\frac{n_1}{N},\frac{n_2}{N}, \ldots, \frac{n_M}{N} \Big)
\end{align}
for the photon distribution vector $|\vn\> = (n_1,n_2,\ldots ,n_M)$ (Fig 1. (b)). Therefore, Definition \ref{kconcurrence} of Fock state coherence $k$-concurrence is analogous to Eq. \eqref{slitkconcurrence} with the replacement of ($d$, $|\P\>$) with ($N$, $|\vn\>$).

\subsection*{The Fock state concurrence sum $C_S$ and the complexity of Fock state BS}\label{algorithm}

In this section, we show that the Fock state concurrence sum $C_S$ defined in the former section plays a crucial role in the complexity of Fock state BS.
To see the relation, we first derive a generalized algorithm for computing the transition amplitudes of Fock state BS with multiple photons in input/output modes. An intriguing fact about our algorithm is that the minimal runtime $\cT_{\min}$ for the algorithm is equal to that of another generalized algorithm presented in Yung $et$ $al$. (2016) \cite{yung2016}. Furthermore, the functional form of $\cT_{\min}$ explicitly contains the Fock state concurrence sum $C_S$. This implies that $C_S$ is a quantum measure that determines the computational complexity of a generalized Fock state BS system.


\subsubsection*{The derivation of a generalized algorithm for matrix permanents}

In the linear optical network of $M$ optical modes characterized by a unitary transformation $\hU$, the photon creation and annihilation operators $\ha^{\dagger}_i$ and $\ha_i$  in the $i$-th mode ($i=1,\ldots ,M$) rotate under the acton of $\hU$ as 
\begin{align}
\hU \ha_i^{\dagger} \hU^\dagger = \sum_{j=1}^MU_{ij}\ha^\dagger_j, \qquad \hU \ha_i \hU^\dagger = \sum_{j=1}^MU^*_{ij}\ha_j.
\end{align}

Scheel \cite{scheel2004} showed that the transition amplitude between the two Fock states $|\vn\>$ ($=\ket{n_1, n_2, ..., n_M}$, $\sum_i^M n_i =N$) and $|\vm\>$ ($=\ket{m_1, m_2, ..., m_M}$, $\sum_i^M m_i=N$ ) is proportional to the matrix permanent:
\begin{align}
\bra{\vm} \hat{U}\ket{\vn}
&=\frac{ \mathrm{Per}\big( [U]_{\vn,\vm}\big)  }{\prod_{k=1}^{M}\sqrt{n_{k}! m_{k}!} }, 
\label{eq:permanent}
\end{align}
where $[U]_{\vn,\vm}$ is an $N\times N$ submatrix of $U$, which has $n_i$ of the $i$-th rows of $U$ and $m_j$ of $j$-th row of $U$. $N= \sum_{i}^{M}n_{i}=\sum_{i}^{M}m_{i}$ is the total number of photons.


The relation \eqref{eq:permanent} holds for the arbitrary square complex matrix $A$ \cite{ma1990}, i.e.,
\begin{align}
\bra{\vm} \hat{A}\ket{\vn}
&=\frac{ \mathrm{Per}\big( [A]_{\vn,\vm}\big)  }{\prod_{k=1}^{M}\sqrt{n_{k}! m_{k}!} }, 
\label{eq:permanentA}
\end{align} where $\hat{A}$ can be expressed as $\hat{A}=\exp[\sum_{i,j}\hat{a}^\dagger_i(\ln A^T)_{ij}\hat{a}_j]$. This  implies that a matrix permanent is obtained by calculating the corresponding transition amplitude between the given input-output Fock states. We will exploit this relation to obtain a generalized algorithm for matrix permanents.

The following lemma \cite{kan:2007} is useful for deriving our algorithm.

\begin{lemma}\label{kan}
	For a vector $\vec{x} =(x_1, x_2,\ldots, x_M)$, the following identity holds:
	\begin{align}
	\label{expansion}
	x_{1}^{n_{1}}\cdots x_{M}^{n_{M}}= =\frac{1}{N!2^{N}}\sum_{v_{1}=0}^{n_{1}}\cdots\sum_{v_{M}=0}^{n_{M}} & (-1)^{N_{v}}
	\begin{pmatrix}
	n_{1}\\v_{1}
	\end{pmatrix}
	\cdots
	\begin{pmatrix} 
	n_{M}\\v_{M}
	\end{pmatrix}
	\Big[\sum_{i=1}^{M} (n_i-2 v_i)x_i\Big]^{N},
	\end{align}
	where  $N_v =\sum_{i=1}^M v_i$ and $N=\sum_{i=1 }^M n_i$.
\end{lemma}
A detailed proof is given in Kan (2007)\cite{kan:2007}. The isomorphism
between quantum states and multivariate polynomials from Theorem 3.6 of Aaronson and Arkhipov (2011) \cite{Aaronson2011} (see also Yung et. al (2016) \cite{yung2016}) connects the above lemma with the linear optical quantum system, which results in the following identity:

\begin{theorem}
	There exists a generalized formula for the matrix permanent with repeated rows and columns that is expressed as
	\begin{align}
	\mathrm{Per}\left([A]_{\vec{n},\vec{m}}\right)=&
	\frac{1}{2^{N}}\sum_{v_{1}=0}^{n_{1}}\cdots\sum_{v_{M}=0}^{n_{M}}(-1)^{N_{v}}
	\begin{pmatrix}
	n_{1}\\v_{1}
	\end{pmatrix}
	\cdots
	\begin{pmatrix}
	n_{M}\\v_{M}
	\end{pmatrix}
	\prod_{j=1}^{M}\Big[\sum_{i=1}^M (n_i-2v_i) A_{ij}\Big]^{m_{j}}. 
	\label{eq:permseries2}
	\end{align}	
\end{theorem}
\begin{proof}
	From Theorem 3.6 of Aaronson and Arkhipov (2011) \cite{Aaronson2011}, the Fock state $|\vn\>=\bigotimes_{k=1}^{M}(\hat{a}^{\dagger}_{k})^{n_{k}}/\sqrt{n_{k}!}\vert 0\rangle$ can be expanded using Lemma \ref{kan} as
	\begin{align}
	\vert\vec{n}\rangle=&\frac{1}{(\prod_{k=1}^{M}\sqrt{n_{k}!})N!2^{N}} \sum_{v_{1}=0}^{n_{1}}\cdots\sum_{v_{M}=0}^{n_{M}}(-1)^{N_{v}} 
	\begin{pmatrix}
	n_{1}\\v_{1}
	\end{pmatrix}
	\cdots
	\begin{pmatrix}
	n_{M}\\v_{M}
	\end{pmatrix}
	\Big[\sum_{i=1}^M(n_i-2v_i)\ha_i^{\dagger}\Big]^{N}
	\vert\vec{0}\rangle .
	\label{eq:Fockseries}
	\end{align}
	By substituting Eq. \eqref{eq:Fockseries} into Eq. \eqref{eq:permanentA}, we have
	\begin{align}
	\mathrm{Per}\left([A]_{\vn,\vm}\right)=
	&\frac{1}{N!2^{N}} \sum_{v_{1}=0}^{n_{1}}\cdots\sum_{v_{M}=0}^{n_{M}}(-1)^{N_{v}}
	\begin{pmatrix}
	n_{1}\\v_{1}
	\end{pmatrix}
	\cdots
	\begin{pmatrix}
	n_{M}\\v_{M}
	\end{pmatrix} 
	\langle 0\vert\hat{a}_{1}^{m_{1}}\cdots\hat{a}_{M}^{m_{M}}
	\Big[ \sum_{ij}(n_i-2v_i) A_{ij} \hat{a}^{\dagger}_j\Big]^{N}
	\vert 0\rangle .
	\label{eq:permseries1}
	\end{align}	
	On the other hand, $\Big[ \sum_{ij}(n_i-2v_i) A_{ij} \hat{a}^{\dagger}_j\Big]^{N}$ is expanded using the multinomial formula 
	\begin{align}
	\Big(\sum_{i=1}^M y_i\Big)^N = \sum_{\sum_{i=1}^M s_i=N}\frac{N!}{\prod_i s_i}y_1^{s_1}\cdots y_M^{s_M},
	\end{align} as
	\begin{align}\label{uexpansion}
	&\Big[ \sum_{ij}(n_i-2v_i) A_{ij} \hat{a}^{\dagger}_j\Big]^{N} =\sum_{\sum_{i=1}^M s_i=N}\frac{N!}{\prod_i s_i!} \prod_{j}\Big[\sum_{i}(n_i-2v_i) A_{ij} \hat{a}^{\dagger}_j\Big]^{s_j}.
	\end{align}
	Therefore, by substituting Eq. \eqref{uexpansion} into Eq. \eqref{eq:permseries1}, we obtain  Eq. \eqref{eq:permseries2} with the identity 
	$\langle 0\vert\hat{a}_{1}^{m_{1}}\cdots\hat{a}_{M}^{m_{M}}(\hat{a}_{1}^{\dagger})^{s_{1}}\cdots(\hat{a}_{M}^{\dagger})^{s_{M}}\vert0\rangle=\prod_{k=1}^{M}(m_{k}!\delta_{m_k s_k})$.
	
\end{proof}

With $A=U$ (unitary operator), the above formula is for the computation of transition amplitudes with multiple photons in both input and output modes.

By exploiting the symmetry in Eq. \eqref{expansion}~\cite{kan:2007}, the number of terms can be reduced to about a half of that in Eq. \eqref{eq:permseries2}. First, when at least one of $\{n_{k}\}$ is an odd number, $n_{1}$ can be chosen to be an 
odd number without loss of generality. Then Eq.~\ref{eq:permseries2} is simplified as
\begin{align}
&\mathrm{Per}\left([A]_{\vn,\vm}\right) =
\frac{1}{2^{N-1}}\sum_{v_{1}=0}^{(n_{1}-1)/2}\cdots\sum_{v_{M}=0}^{n_{M}}  (-1)^{N_{v}}
\begin{pmatrix}
n_{1}\\v_{1}
\end{pmatrix}
\cdots
\begin{pmatrix}
n_{M}\\v_{M}
\end{pmatrix} \prod_{j=1}^{M}\Big[\sum_{i}(n_i-2v_i)A_{ij}\Big]^{m_{j}} . 
\label{eq:permseries3}
\end{align}
Second, when all $\{n_{k}\}$ are even numbers, one still can reduce the number of terms in the summation by dividing the first 
summation of $v_{1}$ into a summation from 0 to $n_{1}-1$ and $v_1=n_1$. Since $n_{1}-1$ is an odd number, the same symmetry that is used for Eq. \eqref{eq:permseries3} reduces the number of terms.  $\prod_{k}(n_{k}+1)/2$ terms are required for the first case and $(\prod_{k}(n_{k}+1)-1)/2$ terms for the second case. 

Now, we show that our formula is reduced to that of Glynn's \cite{Glynn2010, Glynn2013} when $\vn=\vm=(1,\ldots,1)$ and $N=M$. In this case  $(n_i-2v_i)$ is either +1 or -1 for all $i$, 
and all the binomial coefficients become 1. Then Eq. \eqref{eq:permseries3} can be expressed as
\begin{align}
\textrm{Per}(A) = \frac{1}{2^{N-1}}\sum_{\vec{x}\in \{-1,1\}^N}(\prod_{i=1}^N x_i)\prod_{j=1}^N\Big(\sum_{k=1}^N A_{jk}x_k \Big)
\end{align}
which is the Glynn's formula (see Method).

As pointed out by Gurvits~\cite{gurvits2005,aaronson2012}, Glynn's form can be interpreted as 
an ensemble average over the random vector whose entries are $\pm 1$. Likewise, one can interpret Eq. \eqref{eq:permseries2} as a randomized 
algorithm in which $v_{k}$ is randomly generated among $(0,1, \ldots,n_{k})$ with probability $p(v_{k})=\binom{n_{k} }{v_{k} }/2^{n_k}$. Accordingly, Eq.~\ref{eq:permseries2} is rewritten as 
\begin{align}
\mathrm{Per}\left([A]_{\vn,\vn}\right)=
\sum_{\vec{v}=\vec{0}}^{\vn}
p(\vec{v})
G(\vec{v}) ,
\label{eq:permseries6}
\end{align}
where $p(\vec{v})\equiv p(v_{1})\cdots p(v_{M})$ and $\sum_{\vec{v}=\vec{0}}^{\vec{n}}
p(\vec{v})=1$, and $G(\vec{v})\equiv (-1)^{N_{v}}
\prod_{j=1}^{M}[\sum_{i}(n_i-2v_i)A_{ij}]^{m_{j}}$. 
$G(\vec{v})$ is evaluated for each random instance of $\vec{v}$ with the probability $p(\vec{v})$, and then the matrix permanent is approximated as an average, 
\begin{align}
\mathrm{Per}\left([A]_{\vec{n},\vec{m}}\right)\simeq \frac{1}{N_{\mathrm{Sample}}}\sum_{i=1}^{N_{\mathrm{Sample}}}G(\vec{v}^{(i)}),
\end{align}
where $N_{\mathrm{Sample}}$ is the number of samples.

\subsubsection*{Minimal classical runtime} 

The runtime for the classical simulation of Eq. \eqref{eq:permseries2} is obtained by identifying all the summations included in the algorithm as
\begin{align}
\cT(\vn,\vm) = \cO\Big( \prod_{i=1}^{M}(n_i+1)  \a_{\vec{n} } \a_{\vec{m}} \Big) ,
\end{align}
where $\a_\vn$ and $\a_\vm$ are the number of nonzero elements of $\vn$ and $\vm$ respectively. $\prod_{i=1}^{M}(n_i+1)$ comes from $\sum_{v_1=0}^{n_1}\cdots \sum_{v_M=0}^{n_M}$,
and $\a_{\vec{n} } \a_{\vec{m} }$ comes from
$\prod_{j=1}^M [\sum_{i}(n_i -2v_i) A_{ij}]^{m_j}$
in Eq. \eqref{eq:permseries2}.

On the other hand, we can expand $|\vm\>$ instead of $|\vn\>$ with Kan's  series expansion as in Eq. \eqref{eq:Fockseries}. From this input-output symmetry, we obtain another runtime
\begin{align}
\cT'(\vn,\vm) = \cO\Big( \prod_{i=1}^{M}(m_i+1)\a_{\vec{n} } \a_{\vec{m} } \Big).
\end{align}
We can choose the shorter one between $\cT$ and $\cT'$ for the optimal classical simulation. 
Therefore, the minimal running time for the algorithm, denoted by
$\cT_{min}(\vec{n},\vec{m})$, is given by 
\begin{align}
\cT_{min}(\vec{n},\vec{m})=
\mathcal{O}\Big[ \min\Big( \prod_{i=1}^{M}(n_i+1), \prod_{j=1}^{M}(m_j+1) \Big) \a_{\vec{n} } \a_{\vec{m} }  \Big].
\label{rt} 
\end{align}
A special case is when both $n_i$ and $m_i$ are not bigger than $1$ for all $i$. Then $\a_\vn=\a_\vm=N$ and $\cT(\vec{n},\vec{m}) = 2^N N^2$, which is the same runtime as that of Ryser's formula.

The minimal runtime for our algorithm can be compared to that of another generalized algorithm suggested in  Yung $et$ $al.$ (2016) \cite{yung2016}. Interestingly enough, the minimal runtime for the algorithm is exactly equal to that of ours (the same thing happens when we compare the runtime for Ryser's  and Glynn's formula). A brief explanation is presented in Methods. 
This phenomena is intriguing since these two algorithms appear from very different mathematical backgrounds. While our algorithm is constructed from a series expansion of collective variables, the algorithm in Yung $et$ $al.$ (2016) is a direct generalization of Aaronson and Hance's algorithm \cite{aaronson2012}.  Two algorithms created from two totally different paths have the same classical runtime, from which we can surmise that the minimal runtime $\cT_{min}$ is a credible criterion for the computational complexity of the generalized Fock state BS.

\subsubsection*{The minimal runtime and the Fock state concurrence sum}

Now we are ready to see the functional relation of $\cT_{\min}$ with the Fock state concurrence sum $C_S$. Actually, this relation is easily observed by reexpressing $\cT_{\min}(\vn,\vm)$ by expanding $\prod_{i=1}^M (n_i+1)$ along the order of $n_i$ as
\begin{align}
\prod_{i=1}^M (n_i+1)=&  1 + \sum_{i}n_i + \sum_{i_1<i_2}n_{i_1}n_{i_2} + \cdots \nn \\ 
& + \sum_{i_1<i_2< \cdots <i_M}n_{i_1}n_{i_2}\cdots n_{i_M}.
\end{align}
From the definition of the elementary symmetric polynomial (see Eq. \eqref{esp}),
we have 
\begin{align}
\prod_{i=1}^M (n_i+1)= \sum_{k=0}^{\a_\vn} X_{k}(\vn).
\end{align}
Note that the summation is until $k=\a_{\vec{n}}$ because $X_k(\vn) =0$ for $k> \a_\vn$.
As a result, $\cT_{min}(\vn,\vm)$ is rewritten as
\begin{align}
\cT_{min}(\vec{n},\vec{m})=
\mathcal{O}\Big[\Big(\min \Big( \sum_{k=0}^{\a_{\vec{n}} }X_k(\vec{n}), \sum_{l=0}^{\a_{\vec{m}} }X_l(\vec{m})\Big) \a_{\vec{n} } \a_{\vec{m} }  \Big].
\end{align}

By using Definition \ref{concurrencesum}, the minimal runtime is finally rewritten as  
\begin{align}
\cT_{min}(\vec{n},\vec{m})=
\mathcal{O}\Big[2^N \Big(\min[C_S(\vec{n}),C_S(\vec{m}) ] \Big) \times \a_{\vec{n} } \a_{\vec{m} }  \Big],
\label{runtimeconcurrence}
\end{align}
which is a composition of the Fock state concurrence sum and Fock state coherence rank. This expression shows that in linear optics the Fock state concurrence sum is a critical resource that determines the computational complexity. We should emphasize that --in so far as we know-- this is the first evidence that the summation of all the family members of concurrence can operate as an independent resource. Most works on the generalized concurrence in entanglement and coherence have focused on the role of some specific member as the resource for practical quantum processes (see, e.g., Sent{\'\i}s $et$ $al.$ (2016) \cite{sentis2016quantifying}, Girard $et$ $al.$ (2017) \cite{girard2017entanglement}, Chin (2017) \cite{ Chin2017cn} and Chin (2017) \cite{Chin2017gcn}). On the other hand, as we have just seen, the generalized coherence concurrence of Fock state acts as a whole in the multimode linear optical system. In other words, not an individual member $C_k$ but the summation of the whole members $C_S$ becomes the deterministic resource for the process we are interested in.   

As an example, when $\vn = (N, 0,\cdots, 0)$, we have $C_k(\vn)=0$ for $k \ge 2$ and  $C_S(\vn) = \frac{1+N}{2^N}$ (the minimal concurrence sum), which results in $\cT_{\min}(\vn,\vm)=\cO \Big[ (1+N)\a_{\vec{m}} \Big] \le \cO \Big[ (1+N)N\Big]$. For this case the runtime becomes polymonial. As another example, when $\vn=\vm=(1,\cdots, 1, 0,\cdots,0)$, we have $C_S(\vn)=C_S(\vm) = 1$ (the maximal concurrence sum) and $\cT_{\min}(\vn,\vm)=\cO \Big[ 2^N N^2 \Big]$.



Eq. \eqref{runtimeconcurrence} also reveals an intriguing property of $\cT_{min}$, which contrasts with that of the additive error bound $\cE$ for an approximated permanent estimator. In Chin $et$ $al.$ (2017)\cite{Chin2017maj}, \emph{the Boltzmann and Shannon entropy of elementary quantum complexity} is introduced to evaluate the quantum complexity of the given quantum particle distributions. And $\cE$ is explicitly expressed as the difference between the Boltzmann entropy and Shannon entropy of elementary quantum complexity. On the other hand, the relation between the entropies and $\cT_{min}$ is implicit and only can be intuitively explained. Eq. \eqref{runtimeconcurrence} indicates that the generalized Fock state concurrence is another criterion for the computational complexity of linear optical systems. We can state that \emph{both entropy and concurrence are crucial measures (or resources) that directly determines the quantum complexity of linear optical computers}.

Before closing this section, it is worth mentioning the role of the extended definitions for general states (Definition \ref{generalconcurrence}). With such definitions, we can calculate the concurrence sum for arbitary states, including coherent states and thermal states, etc.
As a simple example, when  $\r$ is a thermal state, i.e.,  
\begin{align}
\r =\r^{th}=\sum_{N=0 }^{\infty}\sum_{\sum_i n_i=N}^{\infty}\Big( \prod_{i=1}^M\frac{\bar{n}_i^{n_i}}{(\bar{n}_i+1)^{n_i+1}}\Big) |\vn\>\<\vn|,
\end{align} 
where $\bar{n}_i$ represents the average photon number for each $i$,
we have
\begin{align}
C_k(\r^{th}) = \sum_{N=0 }^{\infty}\sum_{\sum_i n_i=N}^{\infty}\Big( \prod_{i=1}^M\frac{\bar{n}_i^{n_i}}{(\bar{n}_i+1)^{n_i+1}}\Big)C_k(|\vn\>).
\end{align}
Since $C_k(|\vn\>)<1$ except when $|\vn\>$ is maximally coherent, it is easy to see that $C_S(\r^{th})< C_S(|\vn\> = |\vec{1}_{N},\vec{0}_{M-N}\>=1$. We can surmise that the reason why the thermal state BS is simpler to simulate classically than the original BS \cite{rahimi2015} is that the former has less quantum resource, i.e., concurrence sum, than the latter. This viewpoint would be compared more rigorously to that of Rahime-Keshari $et$ $al.$ (2015) \cite{rahimi2015} in the future. We expect that it would reveal the role of concurrence sum in the complexity of various BS systems, such as Gaussian BS \cite{lund2014, rahimi2015, hamilton2017} and Vibronic BS \cite{huh2015v, huh2017v}. 

\section*{Discussion}


We expect our research to develop into two aspects, which are closely relevant to each other. First, the relation of the generalized concurrence $C_{k}(\vn)$ with the exact classical simulation of matrix permanents has many similarities with that of the Boltzmann entropy $S_{B}^q(\vn)$ with the randomized algorithm for approximated permanent computation in Chin $et$ $al.$ (2017)\cite{Chin2017maj}.  By delving into the role of $C_k(\vn)$ and $S_B^q(\vn)$ further, we could formulate the quantum resource theory of Fock state that is a useful tool for understanding the quantum computing power of linear optical computing. Second, our approach to the complexity problem of BS can be compared to that of Rahimi-Keshari $et$ $al.$ (2015)\cite{rahimi2015}, which investigated the role of single-mode nonclassicality in computational complexity. This viewpoint is different from Chin $et$ $al.$ (2017) \cite{Chin2017maj} that focused on the multimode quantum correlation. We expect that there exists a unified theory that embraces the partial interpretations of former works, and the Fock state resource theory including concurrence and entropy is a strong candidate for such a theory.

\section*{Methods}

\subsection*{Glynn's formula and its generalization }\label{A}
Here we briefly introduce Glynn's formula for $N\times N$ matrix permanent computation and its generalization for the permanents of matrices that have repeated rows and columns. 

Glynn's formula with the random variable expectation for $N\times N$ matrix $A$ is given by
\begin{align}
\textrm{Per}(A) = \frac{1}{2^{N-1}}\sum_{\vec{x}\in \{-1,1\}^N}(\prod_{i=1}^N x_i)\prod_{j=1}^N\Big(\sum_{k=1}^N A_{jk}x_k \Big).
\label{Glynn's}
\end{align}
The summation of $\vec{x}$ is over $\vec{x} \in \{ -1,1\}^N$, or $\vec{x} \in \mathcal{X} \equiv \mathcal{R}[2]\times \cdots \times \mathcal{R}[2]$, where $\mathcal{R}[i]$ is a set that consists of the $i$th root of unity. 	 

When $A$ has repeated rows or columns, and the $i$th column (or row) is repeated $n_i$-times, Eq. \eqref{Glynn's} is generalized to \cite{aaronson2012} 
\begin{align}
\textrm{Per}(A) =\sum_{\vec{z}\in \mathcal{X}} v_{\vec{n}}^2(\prod_{i=1}^{N} \bar{z}_i^{n_i}) \prod_{j=1}^N \Big( \sum_{k=1}^N A_{jk}z_k\Big),
\label{genGlynn's}
\end{align}
where $\mathcal{X} \equiv \mathcal{R}[n_1+1]\times \cdots \times \mathcal{R}[n_N+1]$ and $v_{\vn}\equiv \sqrt{\prod_{i=1}^N(n_i!/n_i^{n_i})}$.  The runtime for this algorithm is $\cO\big(\prod_{k=1}^N(n_k+1)\a_{\vec{n}}N\big)$  $=\cO\big(\sum_{k=1}^{\a_{\vec{n}} } X_{k}(\vn) \a_{\vec{n}}N \big)$.

When $A$ has repeated rows and columns, and the $i$th column is repeated $n_i$-times and the $j$th column is repeated $m_j$-times, the above equations are expressed more generally  \cite{yung2016} as 
\begin{align}
\textrm{Per}(A) =\sum_{\vec{z}\in \mathcal{X}} v_{\vec{n}}^2(\prod_{i=1}^{N} \bar{z}_i^{n_i}) \prod_{j=1}^N \Big( \sum_{k=1}^N A_{jk}z_k\Big)^{m_j}.
\label{mgGlynn's}
\end{align} 
Here, $\mathcal{X}$ is the same as that in Eq. \eqref{genGlynn's}.	
From the summation form of Eq. \eqref{mgGlynn's} and the symmetry between rows and columns, it is straightforward to see that the minimal runtime $\cT_{min}(\vn,\vm)$ for Eq. \eqref{mgGlynn's} is equal to Eq. \eqref{rt}.


\section*{Acknowledgements}

This work was supported by Basic Science Research Program through the National Research Foundation of Korea (NRF) funded by the Ministry of Education, Science and Technology (NRF-2015R1A6A3A04059773).

\section*{Author contributions statement}
S. C. and J. H. conceived and designed the experiments, worked on the theory, and wrote the paper.

\section*{Additional information}
 \textbf{Competing Financial Interests} \\
 The authors declare that they have no competing interests.

\end{document}